\documentclass[english, 10pt, a4paper]{article}
\usepackage[verbose,a4paper]{geometry}
\usepackage[T1]{fontenc}
\usepackage[utf8]{inputenc}
\usepackage[english]{babel}
\usepackage{xcolor}		
\definecolor{ForestGreen}{HTML}{005F50} 
\usepackage{amsthm, mathtools, amssymb, nicefrac}
\numberwithin{equation}{section}  
\mathtoolsset{showonlyrefs}
\usepackage{enumitem}
\setlist[enumerate,1]{label=(\roman*), ref=(\roman*)}  
\usepackage[unicode=true, bookmarks=true, bookmarksnumbered=false, bookmarksopen=false, breaklinks=false, pdfborder={0 0 0}, pdfborderstyle={}, backref=page, colorlinks=true, pdfauthor={Alois Pichler}, citecolor=ForestGreen, urlcolor= darkgray, linkcolor= ForestGreen]{hyperref}
\usepackage[numbers]{natbib}  
\theoremstyle{plain}
	\newtheorem{theorem}             {Theorem}[section]
	\newtheorem{corollary}  [theorem]{Corollary}
	\newtheorem{lemma}      [theorem]{Lemma}
	\newtheorem{proposition}[theorem]{Proposition}
\theoremstyle{definition}
	\newtheorem{definition} [theorem]{Definition}
	\newtheorem{example}    [theorem]{Example}
\theoremstyle{remark}
	\newtheorem{remark}     [theorem]{Remark}

\usepackage{newtxtext, newtxmath}	
\usepackage{dsfont}   
\usepackage{microtype}

\usepackage{doi}      
\usepackage{orcidlink}

\DeclareMathOperator{\E}{{\mathds E}}
\DeclareMathOperator{\AVaR}{{\mathsf{AV@R}}}
\DeclareMathOperator{\VaR}{{\mathsf{V@R}}}
\DeclareMathOperator{\sign}{sign}		

\DeclareMathOperator*{\essinf}{ess\,inf}
\DeclareMathOperator*{\esssup}{ess\,sup}
\DeclareMathOperator*{\argmin}{{arg\,min}}

\newcommand{\one}{{\mathds 1}} 		

\definecolor{TUC}{RGB}{0, 95, 80}

\begin{document}
\title{Higher order measures of risk and\\  stochastic dominance}
	\author{Alois Pichler\thanks{Technische Universität Chemnitz, Faculty of Mathematics. 90126 Chemnitz, Germany\protect \\
    Contact: \protect\href{mailto:alois.pichler@math.tu-chemnitz.de}{alois.pichler@math.tu-chemnitz.de}
    \orcidlink{0000-0001-8876-2429} \protect\href{https://orcid.org/0000-0001-8876-2429}{orcid.org/0000-0001-8876-2429}}}
	\maketitle
  \begin{abstract}
    Higher order risk measures are stochastic optimization problems by design, and for this reason they enjoy valuable properties in optimization under uncertainties.
    They nicely integrate with stochastic optimization problems, as has been observed by the intriguing concept of the risk quadrangles, for example.

    Stochastic dominance is a binary relation for random variables to compare random outcomes.
    It is demonstrated that the concepts of higher order risk measures and stochastic dominance are equivalent, they can be employed to characterize the other.
    The paper explores these relations and connects stochastic orders, higher order risk measures and the risk quadrangle.

    Expectiles are employed to exemplify the relations obtained.\medskip


    \noindent\textbf{Classification:} 62G05, 62G08, 62G20
  \end{abstract}

\section{Introduction}
  Risk measures are considered in various disciplines to assess and quantify risk. Similarly to assigning a premium to an insurance contract with random losses after appraising its risk, risk measures assign a number to a random variable, which itself has stochastic outcomes.

  This paper focuses on higher order risk measures, as these risk measures naturally combine with stochastic optimization problems or in ‘learning’ objectives, as they are given as the result of optimization problems.
  In addition, these risk measures relate to the risk quadrangle.

  The paper derives explicit representations of higher order risk measures for general, elementary risk measures in a first main result.
  These characterizations are employed to characterize stochastic dominance relations, which are built on general norms. The second main result is a verification theorem. This is a characterization of higher order stochastic dominance relations, which is numerically tractable.

  For Hölder norms, stochastic dominance relations have been considered for example in \citet{Dupacova2014, Kopa2016, Kopa2023} and \citet{ConsigliDentchevaMaggioni}, in portfolio optimization involving commodities (cf.\ \citet{PichlerWestgaard}), and by \citet{Dentcheva2011} and \citet{MaggioniPflug, MaggioniPflug2016} in a multistage setting.
  The paper employs the characterizations obtained to establish relations for general norms.
  A comparison of these methods is given in \citet{GutjahrPi2013}.
  The paper illustrates these connections for expectiles (\citet{Bellini2016,Bellini2007}) and adds a comparison with other risk measures.

  \paragraph{Outline of the paper.}
	The following Section~\ref{sec:Framework} recalls the mathematical framework for higher order risk measures.
  Section~\ref{sec:Spectral} addresses the higher order risk measure associated with the spectral risks, as these risk measures constitute an elementary building block for general risk measures.
  This section develops the first main result, which is an explicit representation of a spectral risk’s higher order risk measure.
  As a special case, the subsequent Section~\ref{sec:StochasticDominance} links and relates stochastic dominance and higher order risk measures. This section presents the second main result, which allows verifying a stochastic dominance relation by involving only finitely many risk levels.
  The final Section~\ref{sec:Expectile} addresses the expectile and establishes the relations of the preceding sections for this specific risk measure.
  Section~\ref{sec:Summary} concludes.

\section{Mathematical Framework}\label{sec:Framework}
Higher order risk measures are a special instance of \emph{risk measures}, often also termed \emph{risk functionals}.
To introduce and recall their main properties we consider a space~$\mathcal{Y}$ of $\mathbb{R}$\nobreakdash-valued random variables on a probability space with measure~$P$ containing at least all bounded random variables, that is, $L^\infty(P)\subseteq\mathcal{Y}$.
A risk measure then satisfies the following axioms, originally introduced by \citet{Artzner1999}.
\begin{definition}[Risk functional]\label{def:Risk}
  Let~$\mathcal{Y}$ be a space of $\mathbb{R}$\nobreakdash-valued random variables on a probability space $(\Omega, \Sigma, P)$. A mapping $\mathcal{R}\colon\mathcal{Y}\to\mathbb{R}$ is
  \begin{enumerate}[noitemsep,nolistsep]
    \item\label{enu:1}monotone, if $\mathcal{R}(X)\le\mathcal{R}(Y) $, provided that $X\le Y$ almost everywhere;
    \item\label{enu:2}positively homogeneous, if $\mathcal{R}(\lambda\,Y)=\lambda\,\mathcal{R}(Y)$ for all $\lambda>0$;
    \item\label{enu:3}translation equivariant, if $\mathcal{R}(c+Y)=c+\mathcal{R}(Y)$ for all $c\in\mathbb R$;
    \item\label{enu:4}subadditive, if $\mathcal{R}(X+Y)\le \mathcal{R}(X)+\mathcal{R}(Y)$ for all $X$ and $Y\in \mathcal{Y}$.
  \end{enumerate}
  A mapping satisfying~\ref{enu:1}–\ref{enu:4} is called a \emph{risk functional}, or a \emph{risk measure}.
\end{definition}
The risk quadrangle (cf.\ \citet{Rockafellar2013}) interconnects risk measures with the regular measure of deviation, error and regret by
\begin{equation}
	\mathcal{R}(Y)= \inf_{c\in\mathbb{R}}\ c+ \mathcal{V}(Y-c), \label{eq:1}
\end{equation}
where $\mathcal{V}$ is called \emph{regret function} (or \emph{optimized certainty equivalent}, cf.\ \citet{bentalteboulle}).
It follows from the relation~\eqref{eq:1} that $\mathcal{R}$ --~if given as in~\eqref{eq:1}~-- is translation equivariant, i.e, $\mathcal{R}$~satisfies $\mathcal{R}(Y+c)= c+ \mathcal{R}(Y)$ for any $c\in\mathbb{R}$ (cf.~\ref{enu:3} above).
In an economic interpretation, the amount~$c$ in~\eqref{eq:1} corresponds to an amount of cash spent today, while the remaining quantity $Y-c$ is invested and consumed later, thus subject to~$\mathcal{V}$.

The risk functional~$\mathcal{R}$ is positively homogeneous, if the regret function~$\mathcal{V}$ is positively homogeneous. If~$\mathcal{V}$ is not positively homogeneous, then one may consider the positively homogeneous envelope
\[	\mathcal{V}_{ \tilde\beta}(Y)=\inf_{t>0}\ t\left( \tilde\beta+\mathcal{V}\left({Y \over t}\right)\right),
\]
where $\tilde\beta \ge 0$ is a risk aversion coefficient. The combined functional
\begin{align}
	\mathcal{R}_\beta(Y) & =\inf_{c\in\mathbb{R}}c+\mathcal{V}_\beta(Y-c)\nonumber \\
 & =\inf_{\substack{t>0\\ q\in\mathbb{R}}}\ t\left(  \tilde\beta+q+\mathcal{V}\left(\frac{Y}{t}-q\right)\right)\label{eq:2}
\end{align}
is positively homogeneous and translation equivariant (cf.~\ref{enu:2} and~\ref{enu:3}).
The \emph{$\varphi$\nobreakdash-divergence risk measure} is an explicit example of a risk measure, which is defined exactly as~\eqref{eq:2}, cf.\ \citet{DommelPichler}.
\medskip

In what follows, we shall address the reverse question first.
That is, given the risk functional~$\mathcal{R}$, what is the regret functional~$\mathcal{V}$ so that~\eqref{eq:1} holds true?
To this end consider a space  $\mathcal{Y}\subset L^1(P)$ endowed with norm $\|\cdot\|$.
We shall assume the norm to be monotone, that is, $\|X\|\le\|Y\|$ provided that $0\le X\le Y$ almost everywhere.
We associate the following family of risk measure with a given norm.
\begin{definition}[Higher order risk measure]
	Let  $\|\cdot\|$ be a monotone norm on $\mathcal{Y}\subset L^1(P)$ with $\E\one=1$, where $\one(\cdot)=1$ is the identically one function on~$\mathcal{Y}$.
  The \emph{higher order risk measure} at risk level $\beta\in[0,1)$ associated with the norm $\|\cdot\|$ is
  \begin{equation}\label{eq:3}
    \mathcal{R}_\beta^{\|\cdot\|}(Y)=\inf_{t\in\mathbb{R}}\ t+\frac{1}{1-\beta}\|(Y-t)_+\|,
  \end{equation}
  where $\beta\in[0,1)$ is the \emph{risk aversion coefficient} and $x_+\coloneqq\max(0,x)$.

  We shall also omit the superscript and write $\mathcal{R}_\beta$ instead of $\mathcal{R}_\beta^{\|\cdot\|}$ in case the norm is unambiguous given the context.
  We shall demonstrate first that the higher order risk measure is well-defined for any $\beta\ge0$.
\end{definition}
\begin{proposition}
  Let $(\mathcal{Y},\|\cdot\|)$ be a normed space of random variables.
  For the functional~$\mathcal{R}_\beta$ defined in~\eqref{eq:3} it holds that
  \begin{equation}\label{eq:13-4}
    -\|Y\|\le \mathcal{R}_\beta(Y)\le {1 \over 1-\beta} \|Y\|,
  \end{equation}
  so that $\mathcal{R}_\beta(\cdot)$ is indeed well-defined on $(\mathcal{Y},\,\|\cdot\|)$ for every $\beta\in[0,1)$.
\end{proposition}
\begin{proof}
	The upper bound follows trivially from the definition by choosing $t=0$ in the defining equation~\eqref{eq:3}.

  For $t\le0$, it holds that $-t=-Y+(Y-t)\le -Y+ (Y-t)_+$. It follows from the triangle inequality that $-t\le\|Y\|+\|(Y-t)_+\|$ and thus
  \begin{equation}
    -\|Y\|\le t+\|(Y-t)_+\|\quad\text{for all }t\le0.
  \end{equation}
  To establish the relation also for $t\ge 0$, we start by observing the following monotonicity property of the objective in~\eqref{eq:3} in addition: for $\Delta t\ge0$, it follows from the reverse triangle inequality that
  \[	\|Y_+\|-\|(Y-\Delta t)_+ \|\le\|Y_+-(Y-\Delta t)_+\|\le\|\Delta t\,\one\|= \Delta t,
  \]
  where we have used that $0\le Y_+-(Y-\Delta t)_+\le\Delta t$ together with monotonicity of the norm.
  Replacing~$Y$ by $Y-t$ in the latter display gives
  \[	t  +\|(Y-t)_+\|\le t+\Delta t+\|(Y-(t+\Delta t))_+\|;
  \]
  that is, the function $t\mapsto t+\|(Y-t)_+\|$ is non-decreasing, which finally establishes that
  \begin{equation}
    -\|Y\|\le t+\|(Y-t)_+\|\quad\text{for all }t\in\mathbb R.
  \end{equation}
  The lower bound in~\eqref{eq:13-4} thus follows from the latter inequality, as $\mathcal R_0(Y)\le \mathcal R_\beta(Y)$ for any $\beta\ge0$.
\end{proof}
For Hölder spaces (i.e., $L^p(P)$ spaces with $p\ge1$ and norm $\|Y\|_p\coloneqq(\E|Y|^p)^{\nicefrac1 p}$), the higher order risk measure has been introduced in \citet{Krokhmal2007} and studied in \citet{Dentcheva2010}.
\begin{lemma}
  $\mathcal{R}_\beta(\cdot)$ is a risk functional, provided that the norm is monotone. Further, $\mathcal{R}_\beta$ is Lipschitz continuous with respect to the norm, the Lipschitz constant is ${1 \over 1-\beta}$.
\end{lemma}
\begin{proof}
  The assertions~\ref{enu:2}–\ref{enu:4} in Definition~\ref{def:Risk} are straight forward to verify; to verify~\ref{enu:1} it is indispensable to assume that the norm is monotone.

  As for continuity, it follows from subadditivity together with~\eqref{eq:13-4} that
  $\mathcal{R}_\beta(Y)- \mathcal{R}_\beta(Z)\le \mathcal{R}_\beta(Y-Z)\le {1 \over 1-\beta} \|Y-Z\|$, and $|\mathcal{R}_\beta(Y)-\mathcal{R}_\beta(Z)|\le {1 \over 1-\beta} \|Y-Z\|$ after interchanging the roles of~$Y$ and~$Z$. Hence, the assertion.
\end{proof}
Note that the higher order risk measure as defined in~\eqref{eq:3} defines a risk functional based on a norm.
In contrast to this construction, a risk functional~$\mathcal{R}$ defines a norm via
\begin{equation}\label{eq:4}
  Y\|\coloneqq\mathcal{R}(|Y|)
\end{equation}
and a Banach space with $\mathcal{Y}=\left\{ Y\in L^1\colon \mathcal{R}(|Y|)<\infty\right\}$ (cf.\ \citet{Pichler2013a}).
Its natural dual norm for $Z\in \mathcal{Z}\coloneqq \mathcal{Y}^*$ is
\begin{align}
\|Z\|^\ast & \coloneqq\sup\left\{ \E YZ\colon\|Y\|\le 1\right\} \label{eq:5}\\
        & =\sup\left\{ \E YZ\colon\mathcal{R}(|Y|)\le1\right\}.
\end{align}
The following relationship allows defining a regret functional to connect a risk functional~$\mathcal R$ with the risk quadrangle.
\begin{proposition}[Duality]\label{prop:1}
 Let $\mathcal{R}$ be a risk functional with associated norm $\|\cdot\|$ and dual norm $\|\cdot\|^*$.
  For the higher order risk functional it holds that
\begin{align}
	\mathcal{R}_\beta(Y) & =\sup\left\{ \E YZ\colon Z\ge0,\ \E Z=1\text{ and }\|Z\|^*\le\frac{1}{1-\beta}\right\} \label{eq:6}\\
 & =\inf_{t\in\mathbb{R}}\ t+\frac{1}{1-\beta}\|(Y-t)_+\|,\label{eq:1.5}
\end{align}
	where $\beta\in[0,1)$.
\end{proposition}
\begin{remark}
	By the interconnecting formula~\eqref{eq:1}, the higher order risk functional $\mathcal R_\beta^{\|\cdot\|}$ associated with the norm $\|\cdot\|$ is the regret function $\mathcal V(\cdot)\coloneqq \frac1{1-\beta}\|(\cdot)_+\|$.
\end{remark}
\begin{proof}
  It holds by the Hahn--Banach theorem and as $(Y-t)_+\ge0$ that
  \[\frac1{1-\beta}\cdot\|(Y-t)_+\|
    = \sup_{\|Z\|^*\le \frac1{1-\beta}} \E Z(Y-t)_+
      \ge \sup_{\substack{\E Z=1,\ Z\ge0,\\ \|Z\|^*\le \frac1{1-\beta}}} \E Z(Y-t)_+.
 \]
	This establishes the first inequality ‘$\le$’ in~\eqref{eq:1.5} with $t+(Y-t)_+\ge Y$, as
  \begin{align}
    t+\frac1{1-\beta}\cdot\|(Y-t)_+\|
    	&\ge \sup_{\substack{\E Z=1\\Z\ge0,\ \|Z\|^*\le \frac1{1-\beta}}} \E \bigl(t+(Y-t)_+\bigr)Z		\\
		&\ge \sup_{\substack{\E Z=1\\Z\ge0,\ \|Z\|^*\le\frac1{1-\beta}}}\E YZ.
	\end{align}

  As for the converse inequality assume first that $Y$ is bounded. Note, that
  \[	\inf_{t\in\mathbb R} t+\E(Y-t)Z
    = \E YZ+\inf_{t\in\mathbb R} t\cdot(1-\E Z) = \begin{cases}
		\E YZ&	\text{if }\E Z=1,	\\
		-\infty &	\text{else},
  \end{cases}
\]
	so that it follows that
  \[	\sup_{\substack{\E Z=1\\ Z\ge0,\ \|Z\|^*\le \frac1{1-\beta}}} \E YZ
  	= \sup_{\substack{Z\ge0,\\ \|Z\|^*\le \frac1{1-\beta}}} \inf_{t\in\mathbb R} t+\E(Y-t) Z.
  \]
  Further, it holds that $\E YZ= t^*+\E Z(Y-t^*)_+$ for $t^*\le Y$ a.s.\ and thus
	\[	\sup_{\substack{\E Z=1,Z\ge0,\\ \ \|Z\|^*\le\frac1{1-\beta}}} \E YZ
 			= \sup_{\substack{Z\ge0,\\ \|Z\|^*\le\frac1{1-\beta}}} t^*+\E Z(Y-t^*)_+
			= t^* +\frac1{1-\beta} \|(Y-t^*)\|\ge \inf_{t\in\mathbb R} t+\frac1{1-\beta}\|(Y-t)_+\|,
 \]
	thus the desired converse inequality, provided that $Y$ is bounded;
	if $Y$ is not bounded, then there is a bounded~$Y_\varepsilon$ with $Y\le Y_\varepsilon$ ($\varepsilon>0$) and $\|Y_\varepsilon-Y\|<\varepsilon$, so that
  \[	\E Z(Y_\varepsilon-t)_+-\varepsilon\E Z\le\E Z(Y-t)_+\le \E Z(Y_\varepsilon-t)_+,
  \]
  so that we may conclude that~\eqref{eq:1.5} holds for every $Y\in\mathcal{Y}$.
\end{proof}
\begin{example}[Hölder spaces]
  The dual norm of the genuine norm $\|X\|_p\coloneqq(\E|X|^p)^{\nicefrac{1}{p}}$ in the Hölder space $L^p(P)$ is $\|Z\|^*=(\E|Z|^q)^{\nicefrac{1}{q}}$ for the Hölder conjugate exponent with $\frac{1}p+\frac{1}q=1$.
  With Proposition~\ref{prop:1} it follows that
  \begin{align*}
  \mathcal{R}_\beta^{\|\cdot\|_p}(Y) & =\inf_{t\in\mathbb{R}}\ t+\frac{1}{1-\beta}\|(Y-t)_+\|_p\\
  & =\sup\left\{ \E YZ\colon\|Z\|_{q}\le\frac{1}{1-\beta},\ Z\ge0\text{ and }\E Z=1\right\} ,
  \end{align*}
  cf.\ also \citet{ShapiroAlois} and \citet{Pichler2017}.
\end{example}
In what follows, we shall elaborate the higher order risk measure and the associated regret function for specific risk measures, specifically the spectral risk measure.

\section{Higher order spectral risk\label{sec:Spectral}}
By Kusuoka’s theorem (cf.\ \citet{Kusuoka}), every law invariant risk functional can be assembled by elementary risk functionals, each involving the average value-at-risk.

The following section develops the explicit representations of the higher order risk measures associated with spectral risk measures first. The explicit representation then is extended to general risk functionals.
\begin{definition}[Spectral risk measures]\label{def:Spectral}
  The function $\sigma\colon[0,1)\to\mathbb{R}$ is called a \emph{spectral function}, if
  \begin{enumerate}[nolistsep,noitemsep]
    \item\label{item:1} $\sigma(\cdot)\ge0$,
    \item\label{item:ii} $\int_0^1\sigma(u)\,\mathrm{d}u=1$ and
    \item\label{enu:dec} $\sigma(\cdot)$ is non-decreasing.
  \end{enumerate}
  The \emph{spectral risk measure} with spectral function $\sigma$ is
  \[	\mathcal{R}_\sigma(Y)\coloneqq\int_0^1\sigma(u)F_Y^{-1}(u)\,\mathrm{d}u,
  \]
  where
  \[F_Y^{-1}(u)\coloneqq\VaR_u(Y)\coloneqq\inf\left\{ x\in\mathbb{R}\colon P(Y\le x)\ge u\right\} 
  \]
  is the \emph{value-at-risk}, the \emph{generalized inverse} or \emph{quantile function}.
\end{definition}

The higher order risk measure of the spectral risk measure is a spectral risk measure itself.
The following theorem presents the corresponding spectral function explicitly and generalizes \citet{Pflug2000}. The result is central towards the main characterization presented in the next sections.
\begin{theorem}[Higher order spectral risk]\label{thm:Spectral}
  Let $\beta\in[0,1)$ be a risk level. The higher order risk functional of the risk functional $\mathcal{R}_{\sigma}$ with spectral function $\sigma(\cdot)$ has the representation
  \begin{equation}\label{eq:7}
    \inf_{t\in\mathbb R}\ t+\frac{1}{1-\beta}\mathcal{R}_\sigma\bigl((Y-t)_+\bigr)=\mathcal{R_{\sigma_\beta}}(Y),
  \end{equation}
  where $\sigma_\beta$ is the spectral function
 \begin{equation}\label{eq:sigma}
    \sigma_\beta(u)\coloneqq\begin{cases}
      0 & \text{if }u<u_\beta,\\
      \frac{\sigma(u)}{1-\beta} & \text{else};
    \end{cases}
  \end{equation}
  here, $u_{\beta}\in\mathbb{R}$ is the $\beta$\nobreakdash-quantile with respect to the density $\sigma$, that is, the solution of
  \begin{equation}\label{eq:Quantile}
    \int_0^{u_\beta}\sigma(u)\,\mathrm{d}u=\beta,
  \end{equation}
  which is unique for $\beta>0$.
\end{theorem}
\begin{proof}
  We remark first that $\sigma_\beta$ indeed is a spectral function, as $\int_0^1 \sigma_\beta(u)\,\mathrm du= \frac1{1-\beta}\int_{u_\beta}^1\sigma(u)\,\mathrm du=\frac{1-\beta}{1-\beta}=1$ by the defining property~\eqref{eq:Quantile} and~\ref{item:ii} in Definition~\ref{def:Spectral}.
  The quantile $u_\beta$ is uniquely defined for $\beta>0$, as the function $\sigma$ is non-decreasing by~\ref{enu:dec}.
  In what follows we shall demonstrate that the infimum in~\eqref{eq:7} is attained at $t^*\coloneqq F_Y^{-1}(u_\beta)$.
  Note first that
  \[	F_{(Y-t)_+}^{-1}(u)=\begin{cases}
      0 & \text{if }u<F_Y(t),\\
      F_{Y}^{-1}(u)-t & \text{else},
    \end{cases}
  \]
  so that
  \[	\mathcal{R}_\sigma\bigl((Y-t)_+\bigr)=\int_0^1\sigma(u)F_{(Y-t)_{+}}^{-1}(u)\,\mathrm{d}u=\int_{F_Y(t)}^1\sigma(u)\bigl(F_Y^{-1}(u)-t\bigr)\,\mathrm{d}u
  \]
  and
  \begin{equation}\label{eq:8}
    (\mathcal{R}_\sigma)_\beta(Y)=\inf_{t\in\mathbb{R}}\ t+\frac{1}{1-\beta}\int_{F_Y(t)}^1\sigma(u)\bigl(F_Y^{-1}(u)-t\bigr)\,\mathrm{d}u.
  \end{equation}

  Assume first that $t\le t^*$. The inequality $u\le F_Y(t)$ is equivalent to $F_Y^{-1}(u)\le t$ (cf.\ \citet{vdVaart}; this relation of functions $F_Y$ and $F_Y^{-1}$ is occasionally called a \emph{Galois connection}), and thus
  \[	\int_{F_Y(t)}^{F_Y(t^*)}\sigma(u)\bigl(F_Y^{-1}(u)-t\bigr)\,\mathrm{d}u\le0,
  \]
  or equivalently
  \[	\int_{F_Y(t)}^1 \sigma(u)\bigl(F_Y^{-1}(u)-t\bigr)\,\mathrm{d}u\le\int_{F_Y(t^*)}^1\sigma(u)\bigl(F_Y^{-1}(u)-t\bigr)\,\mathrm{d}u.
  \]
  Assume next that $u_\beta\le F_Y(t^*)$, then $\int_{F_Y(t^*)}^1\sigma(u)\,\mathrm{d}u\le1-\text{\ensuremath{\beta}}$ so that
  \[\frac{t-t^*}{1-\beta}\int_{F_Y(t^*)}^1\sigma(u)\,\mathrm{d}u\le t-t^*.
  \]
  Combining the inequalities in the latter displays gives
  \begin{equation}
    t^*+\frac{1}{1-\beta}\int_{F_Y(t^*)}^1\sigma(u)\bigl(F_Y^{-1}(u)-t^*\bigr)\,\mathrm{d}u\le t+\frac{1}{1-\beta}\int_{F_Y(t)}^1\sigma(u)\bigl(F_Y^{-1}(u)-t\bigr)\,\mathrm{d}u\label{eq:9}
  \end{equation}
  and thus the assertion, provided that $u_{\beta}\le F_{Y}(t^{*})$ and $t^*\le t$.

  Conversely, assume that $t\le t^*$. Then the inequality $u\le F_Y(t^*)$ is equivalent to $F_{Y}^{-1}(u)\le t^{*}$ and thus
  \[	\int_{F_Y(t)}^{F_Y(t^*)}\sigma(u)\bigl(F_{Y}^{-1}(u)-t^*\bigr)\,\mathrm{d}u\le0,
  \]
  which is equivalent to
  \[	\int_{F_Y(t)}^{1}\sigma(u)\bigl(F_Y^{-1}(u)-t^*\bigr)\,\mathrm{d}u\le\int_{F_Y(t^*)}^1\sigma(u)\bigl(F_Y^{-1}(u)-t^*\bigr)\,\mathrm{d}u.
  \]
  Assume further that $F_Y(t^*)\le u_\beta$, then $\int_{F_Y(t^*)}^1\sigma(u)\,\mathrm{d}u\ge1-\text{\ensuremath{\beta}}$ so that
  \[	t^*-t\le\frac{t^*-t}{1-\beta}\int_{F_Y(t^*)}^1\sigma(u)\,\mathrm{d}u.
  \]
  Combining the latter inequalities gives
  \begin{equation}\label{eq:10}
    t^*+\frac{1}{1-\beta}\int_{F_Y(t)}^{1}\sigma(u)\bigl(F_Y^{-1}(u)-t^*\bigr)\,\mathrm{d}u\le t+\frac{1}{1-\beta}\int_{F_Y(t^*)}^{1}\sigma(u)\bigl(F_Y^{-1}(u)-t\bigr)\,\mathrm{d}u.
  \end{equation}
  It follows from~\eqref{eq:9} and~\eqref{eq:10} that $t^*\coloneqq F_Y^{-1}(u_\beta)$ is optimal in~\eqref{eq:8}. That is,
  \begin{align}
    (\mathcal{R}_\sigma)_\beta(Y) & =t^{*}+\frac{1}{1-\beta}\int_{u_\beta}^{1}\sigma(u)\bigl(F_Y^{-1}(u)-t^*\bigr)\,\mathrm{d}u\nonumber \\
                                      & =\frac{1}{1-\beta}\int_{u_\beta}^{1}\sigma(u)F_Y^{-1}(u)\,\mathrm{d}u\nonumber \\
                                      & =\int_0^1\sigma_\beta(u)F_Y^{-1}(u)\,\mathrm{d}u\label{eq:Comontone}
  \end{align}
  and thus the assertion.
\end{proof}
The following statement expresses the higher order risk functional by at the base value $u_\beta$, and the random variable’s aberrations to the right, involving the survival function instead of its inverse distribution function.
\begin{corollary}The higher order spectral risk measure is
  \begin{equation}\label{eq:11}
  \bigl(\mathcal{R}_{\sigma}\bigr)_\beta(Y)=\VaR_{u_\beta}(Y)+\frac{1}{1-\beta}\int_{\VaR_{u_\beta}(Y)}^\infty\Sigma\bigl(F_Y(y)\bigr)\mathrm{d}y
  \end{equation}
  (with $u_\beta$ as in~\eqref{eq:Quantile}) or, provided that~$Y$ is bounded,
  \begin{align}\label{eq:12}
  \bigl(\mathcal{R}_\sigma\bigr)_{\beta}(Y) & =\essinf Y+\int_{\essinf Y}^\infty\Sigma_\beta\bigl(F_Y(y)\bigr)\mathrm{d}y,
  \end{align}
  where
  \[	\Sigma_\beta(u)\coloneqq\min\left(1,\ \frac{1}{1-\beta}\int_{u}^{1}\sigma(p)\,\mathrm{d}p\right)
  \]
  is the cumulative spectral function and $\Sigma(u)\coloneqq\Sigma_0(u)=\int_u^1\sigma(p)\,\mathrm{d}p$.
\end{corollary}
\begin{proof}
  Notice first that $\Sigma_\beta(u)=1$ for $u\le u_\beta$, where~$u_\beta$ is given in~\eqref{eq:Quantile}.
  By Theorem~\ref{thm:Spectral}, Riemann–Stieltjes integration by parts and changing the variables it holds that
  \begin{align}
  \bigl(\mathcal{R}_\sigma\bigr)_\beta(Y) & =\mathcal{R}_{\sigma_\beta}(Y)\nonumber \\
  & =\frac{1}{1-\beta}\int_{u_\beta}^1\sigma(u)F_Y^{-1}(u)\,\mathrm{d}u\nonumber \\
  & =-\int_0^1F_Y^{-1}(u)\,\mathrm{d}\Sigma_\beta(u)\label{eq:13}\\
  & =-\left.F_Y^{-1}(u)\Sigma_\beta(u)\right|_{u=0}^1+\int_0^1\Sigma_\beta(u)\,\mathrm{d}F_Y^{-1}(u)\nonumber \\
  & =\essinf Y+\int_{\essinf Y}^\infty\Sigma_\beta\bigl(F_Y(y)\bigr)\,\mathrm{d}y,\label{eq:5-1}
  \end{align}
  where we have used that $F_Y^{-1}(0)=\essinf Y$ and $\Sigma_\beta(1)=0$ in~\eqref{eq:5-1}.
  This gives~\eqref{eq:12}.

  The equation~\eqref{eq:11} results from sticking to the lower bound $u_\beta$ (instead of $0$) in~\eqref{eq:13}. That is,
  \begin{align*}
  \bigl(\mathcal{R}_\sigma\bigr)_\beta(Y) & =-\int_{u_\beta}^1F_Y^{-1}(u)\,\mathrm{d}\Sigma_\beta(u)\\
  & =-\left.F_Y^{-1}(u)\Sigma_\beta(u)\right|_{u=u_\beta}^1+\int_{u_\beta}^1\Sigma_\beta(u)\,\mathrm{d}F_Y^{-1}(u)\\
  & =\VaR_{u_\beta}(Y)+\int_{\VaR_{u_\beta}(Y)}^\infty\Sigma_\beta\bigl(F_Y(y)\bigr)\,\mathrm{d}y,
  \end{align*}
  which is assertion~\eqref{eq:11}.
\end{proof}
\begin{corollary}
  The higher order spectral risk measure has the representation
  \[
  \mathcal{R}_{\sigma_\beta}(Y)=\sup \{ \E[Y\cdot\sigma_\beta(U)]\colon U\in[0,1]\text{ is uniformly distributed} \}.
  \]
\end{corollary}
\begin{proof}
  Recall first that $Y\sim F_Y^{-1}(U)$ for $U$ uniformly distributed.
  By the rearrangement inequality, $\E Y\sigma_\beta(U)\le\E F_{Y}^{-1}(U)\sigma_{\beta}(U)$, because $F_Y^{-1}(U)$ and $\sigma_\beta(U)$ are comonotone and both, $F_Y^{-1}(\cdot)$ and $\sigma_\beta(\cdot)$ are non-decreasing functions.
  The assertion follows with~\eqref{eq:Comontone}.
\end{proof}

The celebrated formula (cf.\ \citet{Pflug2000, RockafellarUryasev2000, RuszOgryczak})
\[  \AVaR_\alpha(Y)= \frac1{1-\alpha}\int_\alpha^1\VaR_u(Y)\,\mathrm{d}u =\inf_{t\in\mathbb{R}}\ t+\frac1{1-\alpha}\E (Y-t)_+
\]
for the average value-at-risk is a special case of preceding Theorem~\ref{thm:Spectral} for the spectral function $\sigma(\cdot)= \frac1{1-\alpha}\one_{[\alpha,1]}(\cdot)$.

The following corollary estabilshes this risk functional’s higher order variant.
\begin{corollary}[Average value-at-risk]\label{cor:AVaR-1}
  The higher order average value-at-risk is
  \begin{equation}\label{eq:14}
    (\AVaR_{\alpha})_\beta(Y)=\AVaR_{1-(1-\alpha)(1-\beta)}(Y),
  \end{equation}
  where $Y\in L^1$; equivalently,
  \begin{equation}\label{eq:15}
    \AVaR_\beta(Y)=\inf_{t\in\mathbb{R}}\ t+\frac{1}{1-\frac{\beta-\alpha}{1-\alpha}}\AVaR_\alpha\bigl((Y-t)_+\bigr),
  \end{equation}
  where $\beta\ge\alpha$.
\end{corollary}
\begin{proof}
  The spectral function of the average value-at-risk is $\sigma_\alpha(\cdot)=\frac{\one_{\cdot\ge\alpha}}{1-\alpha}$.
  It follows from~\eqref{eq:Quantile} that $u_\beta=\alpha+\beta(1-\alpha)=1-(1-\alpha)(1-\beta)$
  and $(\sigma_{\alpha})_{\beta}=\begin{cases} 0 & \text{if }u\le u_{\beta},\\ \frac{1}{(1-\alpha)(1-\beta)} & \text{else}. \end{cases}$ \ This is the spectral function of the average value-at-risk at risk level $u_\beta$, thus the result.

  The assertion~\eqref{eq:15} follows by replacing $\beta$ with $\frac{\beta-\alpha}{1-\alpha}$ in~\eqref{eq:14}.
\end{proof}
\begin{corollary}[Kusoka representation spectral risk measures]\label{cor:mu}Suppose the risk functional is
  \begin{equation}\label{eq:16}
  \mathcal{R}(Y)=\int_0^1\AVaR_\gamma(Y)\,\mu(\mathrm{d}\gamma),
  \end{equation}
  where $\mu$ is a probability measure on $[0,1]$. Then the higher order risk measure is
  \[
  \mathcal{R}_\beta(Y)=\int_0^1\AVaR_\gamma(Y)\,\mu_\beta(\mathrm{d}\gamma),
  \]
  where $\mu_{\beta}(\cdot)$ is the measure
  \begin{equation}\label{eq:mu}
  \mu_\beta(A)\coloneqq p_0\cdot\delta_{u_\beta}(A)+\frac{1}{1-\beta}\mu\bigl(A\cap(u_\beta,1 ]\bigr)
  \end{equation}
  and $u_{\beta}$ and $p_{0}$ are determined by the equation and definition
  \begin{equation}\label{eq:17}
  \int_0^{u_\beta}\frac{u_\beta-\alpha}{1-\alpha}\mu(\mathrm{d}\alpha)=\beta\ \text{ and }\ p_{0}\coloneqq\frac{1-u_\beta}{1-\beta}\int_0^{u_\beta}\frac{\mu(\mathrm{d}\alpha)}{1-\alpha}.
  \end{equation}
\end{corollary}
\begin{proof}
  Above all, $\mu_{\beta}$ is a probability measure, because $p_{0}\ge0$ and
  \begin{align*}
  \mu_\beta([0,1]) & =p_0+\frac{1}{1-\beta}\int_{u_\beta+}^1\mu(\mathrm{d}\alpha)\\
  & =\frac{1}{1-\beta}\int_0^{u_\beta}\frac{1-\alpha-(u_{\beta}-\alpha)}{1-\alpha}\mu(\mathrm{d}\alpha)+\frac{1}{1-\beta}\int_{u_\beta+}^1\mu(\mathrm{d}\alpha)\\
  & =\frac{1}{1-\beta}\int_0^1\mu(\mathrm{d}\alpha)-\frac{\beta}{1-\beta}=1.
  \end{align*}
  The spectral function of the average value-at-risk at risk level~$\alpha$ is $\sigma_\alpha(\cdot)=\frac{\one_{\cdot\ge\alpha}}{1-\alpha}$.
  The quantile condition~\eqref{eq:Quantile} thus is
  \[	\beta=\int_0^1\frac{\max(0,u_\beta-\alpha)}{1-\alpha}\mu(\mathrm{d}\alpha)
  \]
  and thus~\eqref{eq:17}.

  For $u<u_\beta$, the spectral function corresponding to the measure $\mathcal{R}_\beta$ in~\eqref{eq:16} is~$0$, which coincides with~\eqref{eq:sigma}. For $u>u_{\beta}$, the spectral function for $\mathcal{R}_{\beta}$ is
  \begin{align*}
  \frac{p_0}{1-u_\beta}\one_{u\ge u_\beta}+\int_{u_\beta}^1\frac{1}{1-\beta}\frac{\one_{u\ge\alpha}}{1-\alpha}\mu(\mathrm{d}\alpha) & =\frac{1}{1-\beta}\int_{0}^{u_{\beta}}\frac{\one_{u\ge u_\beta}}{1-\alpha}\mu(\mathrm{d}\alpha)+\int_{u_\beta}^1\frac{1}{1-\beta}\frac{\one_{u\ge\alpha}}{1-\alpha}\mu(\mathrm{d}\alpha)\\
  & =\frac{1}{1-\beta}\int_0^1\frac{\mu(\mathrm{d}\alpha)}{1-\alpha},
  \end{align*}
  which is the desired result in light of~\eqref{eq:sigma}.
\end{proof}
In situations of practical interest, the risk measure is often given as finite combination of average values-at-risk at varying levels.
The following corollary addresses this situation explicitly.
\begin{corollary}\label{cor:AVaR}
  Suppose that
  \begin{equation}\label{eq:18}
    \mathcal{R}(Y)=\sum_{i=1}^n p_i\cdot\AVaR_{\alpha_i}(Y)
  \end{equation}
  with $p_i\ge0$,  $\sum_{i=1}^n p_i=1$ and $\alpha_i\in[0,1]$ for $i=1,\dots,n$. Then
  \begin{equation}\label{eq:17-2}
    \mathcal{R}_\beta(Y)= p_0\cdot\AVaR_{u_\beta}(Y)+\sum_{i:\alpha_i>u_\beta}\frac{p_i}{1-\beta}\AVaR_{\alpha_i}(Y),
  \end{equation}
  where $u_\beta$ satisfies $\beta=\sum_{i=1}^np_i\frac{\max(0,u_\beta-\alpha_i)}{1-\alpha_i}$ and $p_0\coloneqq\sum_{i:\alpha_i\le u_\beta}\frac{p_i}{1-\alpha_i}\frac{1-u_\beta}{1-\beta}$.

  For large risk levels $\beta$, specifically if
  \begin{equation}\label{eq:17-4}
    		\beta\ge1-\Bigl(1-\max_{i=1,\dots,n}\alpha_i\Bigr)\cdot\sum_{i=1}^n\frac{p_i}{1-\alpha_i},
  \end{equation}
  the involved risk measure~\eqref{eq:17-2} collapses to the average value-at-risk, it holds that
  \[	\mathcal{R}_\beta(Y)=\AVaR_{1-(1-\tilde\alpha)(1-\beta)}(Y), \]
  where $\tilde\alpha$ is the weighed risk quantile $\tilde\alpha\coloneqq\frac{\sum_{i=1}^n\frac{p_i}{1-\alpha_i}\alpha_i}{\sum_{i=1}^n\frac{p_i}{1-\alpha_i}}$.
\end{corollary}
\begin{proof}
	The result corresponds to the measure $\mu=\sum_{i=1}^np_i\,\delta_{\alpha_i}$ in~\eqref{eq:16}, which is a special case in Corollary~\ref{cor:mu}.

	For $u_\beta\ge\alpha_i$, $i=1,\dots,n$, it holds that
  $\beta=\sum_{i=1}^n p_i\frac{u_\beta-\alpha_i}{1-\alpha_i}= \sum_{i=1}^np_i\frac{1-\alpha_i-(1-u_\beta)}{1-\alpha_i} =1-(1-u_\beta)\sum_{i=1}^n\frac{p_i}{1-\alpha_i}$, so that $u_\beta\ge\max_{i=1,\dots,n}\alpha_i$ is equivalent to~\eqref{eq:17-4}.
	It follows that
  \begin{align}
    u_\beta & =1-\frac{1-\beta}{\sum_{i=1}^n \frac{p_i}{1-\alpha_i}}\label{eq:19} \\
       & =1-(1-\beta)\left(1-\frac{\sum_{i=1}^n\frac{p_i}{1-\alpha_i}-\sum_{i=1}^n\frac{p_i(1-\alpha_i)}{1-\alpha_i}}{\sum_{i=1}^n\frac{p_i}{1-\alpha_i}}\right) \\
       & =1-(1-\beta)(1-\tilde\alpha),
  \end{align}
	and $p_0= \sum\frac{p_i}{1-\alpha_i}\frac{1-u_\beta}{1-\beta}=1$, thus the result with~\eqref{eq:17-2}.
\end{proof}
\begin{remark}
  Corollary~\ref{cor:AVaR-1} is a special case of~\eqref{eq:17-2} in the preceding corollary, as $\tilde\alpha= \alpha$ in this case.
\end{remark}

The following statement generalizes the statements from and provides the higher order risk functional for general risk measures.
\begin{theorem}[Kusuoka representation of higher order risk measures]\label{thm:Kusuoka}
  Let $\mathcal{R}$ be a law invariant risk measure with Kusuoka representation
  \begin{equation}
  \mathcal{R}(Y)=\sup_{\mu\in\mathcal{M}}\mathcal{R}_\mu(Y).\label{eq:Kusuoka}
  \end{equation}
  The higher order risk measure is
  \[	\mathcal{R}_\beta(Y)=\sup_{\mu\in\mathcal{M}}\mathcal{R}_{\mu_\beta}(Y),
  \]
  where the truncated measures $\mu_\beta$ are given in~\eqref{eq:mu}.
\end{theorem}
\begin{proof}
  For the risk functional defined in~\eqref{eq:Kusuoka} it follows from the min-max inequality that
  \begin{align}
  \mathcal{R}_\beta(Y) & =\inf_{t\in\mathbb{R}}t+\sup_{\mu\in\mathcal{M}}\mathcal{R}_\mu\bigl((Y-t)_+\bigr)\nonumber \\
  & \ge\sup_{\mu\in\mathcal{M}}\inf_{t\in\mathbb{R}}\ t+\frac1{1-\beta}\mathcal{R}_\mu\bigl((Y-t)_+\bigr)\label{eq:18-1}\\
  & =\sup_{\mu\in\mathcal{M}}(\mathcal{R}_\mu)_\beta(Y)=\nonumber \\
  & =\sup_{\mu\in\mathcal{M}}\mathcal{R}_{\mu_\beta}(Y),\nonumber 
  \end{align}
  where we have used Corollary~\ref{cor:mu}.

  For the reverse inequality in~\eqref{eq:18-1} consider the function
  \[	(t,\mu)\mapsto t+\mathcal{R}_\mu\bigl((Y-t)_+\bigr)
  \]
  on $\mathbb{R}\times \mathcal{M}([0,1])$, where $\mathcal M([0,1])$ collects the probability measures on $[0,1]$ (with its Borel $\sigma$\nobreakdash-algebra).
  By its definition~\eqref{eq:16}, this function is linear in $\mu$, and convex in $t$, where $t\in\mathbb{R}$ and $\mu$ is a measure on $[0,1]$.
  By Prokhorov’s theorem, the set~$\mathcal{M}([0,1])$ of probability measures is sequentially compact, as $[0,1]$ is compact.
  From Sion’s minimax theorem (cf.\ \citet{Sion}) it follows that equality holds in~\eqref{eq:18-1}. Thus, the result.
\end{proof}
The preceding Theorem~\ref{thm:Kusuoka} provides an explicit characterization for the general higher order risk measure.
The following section exploits this representation to characterize general stochastic dominance relations.

\section{General stochastic dominance relations}\label{sec:StochasticDominance}
As Section~\ref{sec:Framework} mentions, the risk measure $\mathcal R$ defines a norm via the setting $\|\cdot\|\coloneqq \mathcal R(|\cdot|)$ (cf.~\eqref{eq:4}), and conversely, the norm $\|\cdot\|$ defines a risk measure via $\mathcal R_\beta^{\|\cdot\|}$, cf.~\eqref{eq:3}.
In what follows we connect  a specific stochastic dominance relation with the norm.
This stochastic dominance relation can be described by higher order risk measures, developed in the preceding Section~\ref{sec:Spectral}.

We start by defining the stochastic dominance relation based on a monotone norm.
\begin{definition}[Stochastic dominance]\label{def:Dominance}
  Let $X$, $Y\in\mathcal{Y}$ be random variables in a Banach space $(\mathcal{Y},\|\cdot\|)$. The random variable~$X$ is \emph{dominated} by $Y$, denoted
  \[ X\preccurlyeq^{\|\cdot\|}Y,
  \]
  if
  \begin{equation}\label{eq:domain}
    \|(t-X)_+\|\ge\|(t-Y)_+\|\ \text{ for all }t\in\mathbb{R}.
  \end{equation}
  If the norm is unambiguous from the context, we shall also simply write $\preccurlyeq$ instead of $\preccurlyeq^{\|\cdot\|}$.
\end{definition}

The cone of random variables triggered by a single variable is convex.
\begin{lemma}[Convexity of the stochastic dominance cone]
  For $X\in\mathcal{Y}$ given, the set
  \[	\left\{ Y\colon X\preccurlyeq Y\right\}
  \]
  is convex.
\end{lemma}
\begin{proof}
  The map $y\mapsto(t-y)_+$ is convex, as follows from reflecting and translating the convex function $x\mapsto x_+$.
  Suppose that $X\preccurlyeq Y_0$ and $X\preccurlyeq Y_1$. Then it follows for $Y_\lambda\coloneqq(1-\lambda)Y_0+\lambda\,Y_1$, together with monotonicity of the norm and~\eqref{eq:domain}, that
  \begin{align*}
  \|(t-Y_\lambda)_+\| & \le\left\Vert \bigl((1-\lambda)(t-Y_0)+\lambda(t-Y_{1})\bigr)_{+}\right\Vert \\
  & \le(1-\lambda)\|(t-Y_0)_+\|+\lambda\|(t-Y_1)_+\|\\
  & \le(1-\lambda)\|(t-X)_+\|+\lambda\|(t-X)_+\|\\
  & =\|(t-X)_+\|.
  \end{align*}
  That is, it holds that $X\preccurlyeq Y_\lambda$ and thus the assertion.
\end{proof}

\subsection{Characterization of stochastic dominance relations\label{sec:Characterization}}
Stochastic dominance relations can be fully characterized by higher order risk measures.
The following theorem presents this main result, which integrates the details developed above for these risk functionals and stochastic dominance relations.
\begin{theorem}[Characterization of stochastic dominance]\label{thm:3}The following are equivalent:
  \begin{enumerate}
    \item\label{enu:6}$X\preccurlyeq^{\|\cdot\|}Y$,
    \item\label{enu:5}$\mathcal{R}_\beta(-X)\ge\mathcal{R}_\beta(-Y)$ for all $\beta\in[0,1)$, and
    \item\label{enu:7}$\inf_{Z\in\mathcal{Z}_\beta}\E ZX\le\inf_{Z\in\mathcal{Z}_\beta}\E ZY$ for every $\beta\in(0,1)$, where
      \begin{equation}
        \mathcal{Z}_\beta
        	\coloneqq
       	 \left\{ 	Y\in\mathcal{X}^*\colon\|Z\|_*\le\frac1{1-\beta},\ \E Z=1,\ Z\ge0\right\}
      \end{equation}
    is the positive cone ($Z\ge0$) in the dual ball with radius $\frac{1}{1-\beta}$ ($\|Z\|_\ast\le \frac1{1-\beta}$), intersected with the simplex ($\E Z=1$).
  \end{enumerate}
\end{theorem}
\begin{proof}
  Suppose that $X\preccurlyeq^{\|\cdot\|}Y$, then, by definition, $\|(t-X)_+\|\ge\|(t-Y)_+\|$ for every $t\in\mathbb{R}$. It follows that $t+\frac{1}{1-\beta}\|(-X-t)_+\|\ge t+\frac{1}{1-\beta}\|(-Y-t)_+\|$ for all $t\in\mathbb{R}$, and thus assertion~\ref{enu:5} after passing to the infimum.

  As for the contrary, assume that~\ref{enu:5} holds. To demonstrate~\ref{enu:6} note first that $q\mapsto\|(q-X)_+\|$ is convex; indeed, with $q_\lambda\coloneqq(1-\lambda)q_0+\lambda\,q_1$ and $(a+b)_+\le a_+ +b_+$ it holds that
  \[	(q_\lambda-X)_+= \bigl((1-\lambda)(q_0-X)+\lambda(q_1-X)\bigr)_+\le(1-\lambda)(q_0-X)+\lambda(q_1-X)_+
  \]
  and thus
  \[	\|(q_\lambda-X)\|\le(1-\lambda)\cdot\|(q_0-X)_+\|+\lambda\cdot\|(q_1-X)_+\|
  \]
  by the triangle inequality of the norm.

  For $q\in\mathbb{R}$ fixed, choose
  \[\alpha\in\partial_\eta\ \|(\eta-Y)_+\|\Big|_{\eta=q},
  \]
 that is, the subdifferential (of the convex function $\eta\mapsto \|(\eta-Y)_+\|$) evaluated at $\eta=q$, and note that $\alpha\in[0,1]$. Set $\beta\coloneqq 1-\alpha$, and observe that
  \[	0\in\partial_q\ -q+\frac1{1-\beta}\|(q-Y)_+\|
  \]
  by~\eqref{eq:3} so that
  \[	\mathcal{R}_\beta(-Y)=-q+\frac{1}{1-\beta}\|(q-Y)_+\|.
  \]
  Employing the definition~\eqref{eq:3} again and assumption~\ref{enu:5}, it follows that
  \begin{align*}
	  -q+\frac{1}{1-\beta}\|(-X+q)_+\| & \ge\mathcal{R}_{\beta}(-X)\\
  & \ge\mathcal{R}_\beta(-Y)\\
  & =-q+\frac{1}{1-\beta}\|(q-Y)_+\|,
  \end{align*}
  or equivalently
  \[	\|(q-X)_+\|\ge\|(q-Y)_+\|.
  \]
  The assertion~\ref{enu:6} follows, as $q\in\mathbb{R}$ was arbitrary; this establishes equivalence of~\ref{enu:6} and~\ref{enu:5}.

  Finally, let $\beta\in(0,1)$. With~\ref{enu:5} and Proposition~\ref{prop:1} we have that
  \[	\inf_{Z\in\mathcal{Z}_\beta}\E ZX\le\inf_{Z\in\mathcal{Z}_\beta}\E ZY,
  \]
  where the infimum in both expressions is among $Z\in\mathcal{Z}_\beta= \left\{ Z\in\mathcal{Z}\colon\ \|Z\|_*\le\frac{1}{1-\beta}\right\} $, as the set $\mathcal Z_\beta$ collects the constraints in~\eqref{eq:6}.
  This establishes equivalence between~\ref{enu:5} and~\ref{enu:7}.
\end{proof}
\begin{remark}
  The quantity $-\mathcal{R}(-Y)\eqqcolon\mathcal{A}(Y)$ arising naturally in Theorem~\ref{thm:3}~\ref{enu:5} above is often called an \emph{acceptability functional}, cf.\ \citet{PflugRomisch2007}.
\end{remark}

\begin{corollary}
  Suppose that
  \begin{equation}
	  \E ZX\le\E ZY\text{ for all }Z\in\mathcal{Z}\coloneqq\bigcup_{\beta\in(0,1)}\mathcal{Z}_{\beta},\label{eq:20}
  \end{equation}
  then $X$ is dominated by $Y$, $X\preccurlyeq^{\|\cdot\|}Y$. Further, the assertion~\eqref{eq:20} is equivalent to
  \begin{equation}\label{eq:21}
	  \mathcal{R}_\beta(X-Y)\le 0\quad\text{for all }\beta\in(0,1).
  \end{equation}
\end{corollary}
\begin{proof}
  Fix $\beta\in(0,1)$, then $\inf_{Z\in\mathcal{Z}_\beta}\E ZX\le\inf_{Z\in\mathcal{Z}_\beta}\E ZY$ by~\eqref{eq:20}. With~\ref{enu:7} in the preceding Theorem~\ref{thm:3} it follows that $X\preccurlyeq Y$.

  With~\eqref{eq:5}, the statement~\eqref{eq:21} is equivalent with $\E Z(X-Y)\le0$ for $Z\in\mathcal{Z}$ and hence the assertion.
\end{proof}
\begin{remark}
  The assertion~\eqref{eq:21}, however, is \emph{strictly stronger} than~\ref{enu:5} in Theorem~\ref{thm:3}. Indeed, it follows with convexity and~\eqref{eq:21} that
  \[	\mathcal{R}(-Y)\le\mathcal{R}(X-Y)+\mathcal{R}(-X)\le\mathcal{R}(-X), \]
  and hence~\ref{enu:5}, the assertion, although the reverse implication does not hold true.
\end{remark}

\subsection{Stochastic dominance in numerical computations}\label{sec:662}
  To verify that $X\preccurlyeq Y$ it is necessary to verify the defining condition~\eqref{eq:domain} for every $t\in \mathbb R$. These infinitely many comparisons are intractable for numerical computations.
  The same holds true for the equivalent characterization~\ref{enu:5} in Theorem~\ref{thm:3},
  as all risk levels $\beta\in (0,1)$ -- again infinitely many -- need to be considered. This is difficult, perhaps impossible to ensure in numerical computations.

  In what follows we develop an equivalent characterization, which builds on only \emph{finitely many} risk levels. With this,  the comparison $X\preccurlyeq Y$ is numerically tractable.

  We start the exposition with the following lemma on convexity (concavity).
\begin{lemma}\label{lem:150}
    For~$Y$ fixed in the domain of $\mathcal R$, the mapping \[\beta\mapsto (1-\beta)\cdot\mathcal R_\beta(Y)\] is concave.
\end{lemma}
\begin{proof}
  For $\lambda\in(0,1)$, define $\beta_\lambda\coloneqq (1-\lambda)\beta_0+\lambda\,\beta_1$. Choose $t_\lambda$ in~\eqref{eq:3} minimizing $\mathcal R_{\beta_\lambda}(Y)$. Then
  \begin{align}
    (1-\beta_\lambda)\mathcal R_{\beta_\lambda}(Y)  &= (1-\beta_\lambda) t_\lambda + \|(Y-t_\lambda)_+\| \\
                                &= (1-\lambda)\bigl((1-\beta_0)t_\lambda+ \|(Y-t_\lambda)_+\|\bigr) +\lambda\bigl((1-\beta_1)t_\lambda+\|(Y-t_\lambda)_+\|\bigr) \\
                                & \ge(1-\lambda)(1-\beta_0)\,\mathcal R_{\beta_0}(Y)+ \lambda(1-\beta_1)\,\mathcal R_{\beta_1}(Y)
  \end{align}
  by taking the infimum in $\mathcal R_\beta$ in the latter expressions.
\end{proof}
The latter result on convexity leads to the following result on the derivative of the risk functional with respect to the risk level.
\begin{theorem}\label{thm:673}
  For $Y$ in the domain of $\mathcal R$ and $\beta\in(0,1)$, the derivative with respect to the risk rate is
  \[ {\mathrm d \over \mathrm d\beta} \mathcal R_\beta(Y)= \frac{\mathcal R_\beta(Y)-t_Y(\beta)}{1-\beta}, \]
  where $t_Y(\beta)$ minimizes the higher order risk measure $\mathcal R_\beta(Y)$, cf.~\eqref{eq:3}.
\end{theorem}
\begin{proof}
    As in Lemma~\ref{lem:150} above consider the objective
    \begin{equation}\label{eq:22}
    f(\beta,t)\coloneqq (1-\beta)t+\|(Y-t)_+\|.
    \end{equation}
      For $\beta$ fixed, the objective is concave, and we may choose $t(\beta)$ in the subgradient, $t(\beta)\in \partial_t f(\beta,t)$ .
      Define the function \[ f(\beta)\coloneqq f\bigl(\beta,t(\beta)\bigr)= \min_{t\in\mathbb R}\  f(\beta,t(\beta)) = (1-\beta)\,\mathcal R_\beta(Y).\]
      It holds that
      \[ f^\prime(\beta)= {\partial \over \partial\beta} f\bigl(\beta,t(\beta)\bigr)+ {\partial \over \partial t} f\bigl(\beta,t(\beta)\bigr)\cdot t^\prime(\beta).\]
      But $0\in {\partial \over \partial t} f\bigl(\beta,t(\beta)\bigr)$, as $t(\beta)$ is optimal in~\eqref{eq:3}.
      With ${\partial \over \partial\beta} f(\beta,t)= -t$, it follows that
      \begin{equation}\label{eq:44}
        f^\prime(\beta)= -t(\beta).
      \end{equation}
      It follows that \[{\mathrm d \over \mathrm d\beta}\mathcal R_\beta(Y)={\mathrm d\over \mathrm d\beta} {f(\beta) \over 1-\beta}= {-t(\beta)(1-\beta)+ f(\beta) \over (1-\beta)^2}
        = {\mathcal R_\beta(Y)-t(\beta) \over 1-\beta},\]
      the assertion.
\end{proof}
\begin{theorem}[Verification of stochastic dominance relations]\label{thm:140}
  To verify that $X\preccurlyeq Y$ it is sufficient to verify
  \begin{align}\label{eq:23}   \mathcal R_{\beta_i}(-X)\ge \mathcal R_{\beta_i}(-Y)
  \shortintertext{for the (\emph{finitely many}) risk levels}
     \beta_i,\quad i=1,\dots, n;
    \end{align}
  the risk levels $\beta_i$  with $\beta_i< \gamma_i< \beta_{i+1}$ are chosen so that
  \begin{align}
    t_{-X}(\beta)&\le t_{-Y}(\beta) \text{ for } \beta\in (\beta_i,\gamma_i) \text{ and}\label{eq:24}\\
     t_{-X}(\beta)&\ge t_{-Y}(\beta) \text{ for } \beta\in (\gamma_i,\beta_{i+1})\label{eq:25}
  \end{align}
  for $i=1,2,\dots,n$.
\end{theorem}
\begin{remark}\label{rem:217}
  For the crucial risk levels (test points) $\beta_i$ in the preceding Theorem~\ref{thm:140} it holds that
  \[  t_{-X}(\beta)\ge t_{-Y}(\beta) \text{ for }\beta\in (\gamma_{i-1},\beta_i)\text{ and } t_{-X}(\beta)\le t_{-Y}(\beta)\text{ for }\beta\in (\beta_i,\gamma_i),\]
  so that the curves $t_{-X}(\cdot)$ and $t_{-Y}(\cdot)$ intersect at~$\beta_i$.
  For continuous $t_{-X}(\cdot)$ and $t_{-Y}(\cdot)$, the crucial points in the preceding Theorem~\ref{thm:140} are among the points, where the minimizers of $\mathcal R_\beta(-X)$ and $\mathcal R_\beta(-Y)$ coincide, that is,
  \[ t_{-X}(\beta_i) = t_{-Y}(\beta_i), \quad i=1,\dots,n;\]
  with a similar argument it holds as well that
  \[ t_{-X}(\gamma_i) = t_{-Y}(\gamma_i), \quad i=1,\dots,n.\]
\end{remark}

\begin{proof}[Proof of Theorem~\ref{thm:140}]
  Recall first from Theorem~\ref{thm:3}~\ref{enu:5} that $X\preccurlyeq Y$ is equivalent to
  \[\mathcal R_\beta(-X)\ge\mathcal R_\beta(-Y)\text{ for all }\beta\in(0,1).\]

  Assuming~\eqref{eq:23}, we demonstrate that $\mathcal R_\beta(-X) \ge \mathcal R_\beta(-Y)$ for $\beta\in (\beta_i,\beta_{i+1})$ for $i=1,\dots,n$.
  To this end we distinguish the following two cases, where we reuse the notation introduced in the proof of Theorem~\ref{thm:673}:
  \begin{enumerate}
    \item Suppose that $\beta\in (\beta_i,\gamma_i)$, then, by~\eqref{eq:24}, $t_{-X}(\beta)\le t_{-Y}(\beta)$.
          With~\eqref{eq:44} we conclude that
       \begin{align}
      f_{-X}(\beta)&= f_{-X}(\beta_i)+\int_{\beta_i}^\beta f_{-X}^\prime(\gamma)\,\mathrm d\gamma\\
                &=  f_{-X}(\beta_i)-\int_{\beta_i}^\beta t_{-X}(\gamma)\,\mathrm d\gamma\\
                &\ge  f_{-Y}(\beta_i)-\int_{\beta_i}^\beta t_{-Y}(\gamma)\,\mathrm d\gamma  \label{eq:26}\\
                &= f_{-Y}(\beta_i)+\int_{\beta_i}^\beta f_{-Y}^\prime(\gamma)\,\mathrm d\\
                &= f_{-Y}(\beta),
    \end{align} where we have used~\eqref{eq:23} in~\eqref{eq:26}. It follows that $\mathcal R_\beta(-X)\ge \mathcal R_\beta(-Y)$ for all $\beta\in(\beta_i,\gamma_i)$.
    \item If $\beta\in (\gamma_i,\beta_{i+1})$, then $t_{-X}(\beta)\ge t_{-Y}(\beta)$ by assumption~\eqref{eq:25}. It holds that
    \begin{align}
      f_{-X}(\beta) &= f_{-X}(\beta_{i+1})- \int_\beta^{\beta_{i+1}}f^\prime_{-X}(\gamma)\mathrm d\gamma\\
      &= f_{-X}(\beta_{i+1})+ \int_\beta^{\beta_{i+1}}t_{-X}(\gamma)\mathrm d\gamma \\
      &\ge f_{-Y}(\beta_{i+1})+ \int_\beta^{\beta_{i+1}}t_{-Y}(\gamma)\mathrm d\gamma \label{eq:27}\\
      &= f_{-Y}(\beta_{i+1})- \int_\beta^{\beta_{i+1}}f^\prime_{-Y}(\gamma)\mathrm d\gamma \\
      &= f_{-Y}(\beta),
    \end{align}
    where again~\eqref{eq:23} was used in~\eqref{eq:27}. It follows that $\mathcal R_\beta(-X)\ge\mathcal R_\beta(-Y)$ for $\beta\in(\gamma_i,\beta_{i+1})$, the remaining case.
  \end{enumerate}
  Combining the two cases above we find that $\mathcal R_\beta(-X)\ge \mathcal R_\beta(-Y)$ for all $\beta\in (\beta_i,\beta_{i+1})$.
  The assertion for all $\beta\in \mathbb R$ thus follows by considering $i=1,\dots,n$.
\end{proof}
\begin{remark}\label{rem:759}
  It is important to note that the critical risk levels $\beta_i$, $i=1,\dots,n$, in~\eqref{eq:23}~-- in general~-- depend on \emph{both} random variables, on $X$ \emph{and} $Y$.
  The folloiwing remark presents a notable exception to this rule.
\end{remark}
\begin{remark}[Average value-at-risk]\label{rem:225}
  For the Hölder norm $\|\cdot\|_1$ with $p=1$, the optimizers are $t_{-X}(\beta)=\VaR_\beta(-X)$, and the function $t_{-X}(\cdot)$ is not continuous for discrete distributions. As $\beta\mapsto t_{-X}(\beta)$ is non-decreasing, the test points are
  \[  \beta_i= -x_i,\ i=1,2,\dots,n,\] where $P(X=x_i)>0$.
  In this case, the crucial risk levels $\beta_i$, $i=1,\dots,n$, are independent from the variable $Y$ when comparing $X\preccurlyeq Y$.
\end{remark}

\subsection{Characterization of stochastic dominance for spectral risk measures}
The following builds on the spectral risk measure $\mathcal R_\sigma(\cdot)$ introduced in Definition~\ref{def:Spectral} and considers the norm
\[	\|\cdot\|_\sigma \coloneqq \mathcal R_\sigma(|\cdot|)
\] for the spectral function $\sigma$.
Theorem~\ref{thm:3} and the characterization of higher order spectral risk measures (Theorem~\ref{thm:Spectral}) give rise to the following result.
\begin{theorem}
  The stochastic dominance relation
  \[
  X\preccurlyeq^{\|\cdot\|_\sigma}Y
  \]
  with respect to the norm associated with the spectral risk measure
  $\mathcal{R}_\sigma$ is equivalent to
  \begin{align*}
  \MoveEqLeft[7]-\sigma_p\VaR_p(Y)+\int_{-\infty}^{\VaR_p(Y)}\Sigma\bigl(S_Y(y)\bigr)\,\mathrm{d}y\\
  & \le-\sigma_p\VaR_p(X)+\int_{-\infty}^{\VaR_p(X)}\Sigma\bigl(S_X(x)\bigr)\,\mathrm{d}x\quad \text{for all }p\in(0,1),
  \end{align*}
  where $\sigma_p\coloneqq\int_{1-p}^1 \sigma(u)\,\mathrm{d}u$ and
  $S_X(x)\coloneqq1-F_X(x)= P(X>x)$ is the \emph{survival function} of the random variable~$X$.
\end{theorem}
\begin{proof}
  We argue with the norm~$\|Y\|_\sigma\coloneqq\mathcal{R}_\sigma(|Y|)$.
  Note, that $(Y-t)_+\ge0$, hence the defining equation~\eqref{eq:3} is
  \begin{align}
    \mathcal{R}_\beta^{\|\cdot\|_\sigma}(Y) & =\inf_{t\in\mathbb{R}}  \ t+\frac{1}{1-\beta}\|(Y-t)_+\|_{\sigma}\\
                                                & =\inf_{t\in\mathbb{R}}\ t+\frac{1}{1-\beta}\mathcal{R}_\sigma\bigl((Y-t)_+ \bigr)\\
                                                & =\mathcal{R}_{\sigma_\beta}(Y),\label{eq:28}
  \end{align}
  where we have used Theorem~\ref{thm:Spectral} in~\eqref{eq:28}.

  From~\eqref{eq:11} we have that
  \begin{align*}
    \mathcal{R}_{\beta}(-Y) & =\VaR_{u_\beta}(-Y)+\frac{1}{1-\beta}\int_{\VaR_{u_\beta}(-Y)}^{\infty}\Sigma\bigl(F_{-Y}(y)\bigr)\,\mathrm{d}y\\
                            & =-\VaR_{1-u_\beta}(Y)+\frac{1}{1-\beta}\int_{-\VaR_{1-u_\beta}(Y)}^{\infty}\Sigma\bigl(S_Y(-y)\bigr)\,\mathrm{d}y\\
                            & =-\VaR_{1-u_\beta}(Y)+\frac{1}{1-\beta}\int_{-\infty}^{\VaR_{1-u_\beta}(Y)}\Sigma\bigl(S_Y(y)\bigr)\,\mathrm{d}y,
  \end{align*}
  where we have used that $F_{-Y}(y)=P(-Y\le y)=P(Y\ge-y)= 1-F_Y(-y)=S_Y(-y)$ and $\VaR_\alpha(-Y)=-\VaR_{1-\alpha}(Y)$ at points of continuity of $F_Y(\cdot)$.

  Now set $1-u_{\beta}\eqqcolon p$. Then, by employing the characterizing relation~\eqref{eq:Quantile} for the $\beta$\nobreakdash-quantile of~$\sigma$, it holds that
  \[	1-\beta=\int_{u_\beta}^1 \sigma(u)\,\mathrm{d}u= \int_{1-p}^1 \sigma(u)\,\mathrm{d}u= \sigma_p, \]
  so that
  \[	\mathcal{R}_\beta(-Y)=-\VaR_{p}(Y)+\frac1{\sigma_p}\int_{-\infty}^{\VaR_p(Y)}\Sigma\bigl(S_Y(y)\bigr)\,\mathrm{d}y. \]

  By Theorem~\ref{thm:3}, the relation $X\preccurlyeq^{\|\cdot\|_\sigma}Y$ is equivalent to $\mathcal{R}_\beta^{\|\cdot\|_\sigma}(-Y)\le\mathcal{R}_\beta^{\|\cdot\|_\sigma}(-X)$ for all $\beta\in(0,1)$. With that, the assertion follows.
\end{proof}

\subsection{Higher order stochastic dominance\label{sec:Hoelder}}
A traditional way of introducing stochastic dominance relations is by iterating integrals of the cumulative distribution function. This is a special case for the Hölder norm $\|\cdot\|_p$, $p\in[1,\infty)$.
\begin{definition}[Higher order stochastic dominance, cf.\ \citet{StoyanMueller2002}]
  The random variable~$X$ is dominated by~$Y$ in \emph{first order stochastic dominance}, if
  \[	F_X(x)\ge F_Y(x)\text{ for all }x\in\mathbb{R},
  \]
  where $F_X(x)\coloneqq P(X\le x)$ is the cumulative distribution function. We shall write $X\preccurlyeq^{(1)}Y$. For $p\in[1,\infty]$, the random variable~$X$ is stochastically dominated by $Y$ in \emph{$p$}\textsuperscript{\emph{th}}\emph{\nobreakdash-stochastic order}, if
  \begin{equation}\label{eq:29}
      \E(x-X)_+^{p-1}\ge\E(x-Y)_+^{p-1}\text{ for all }x\in\mathbb{R};
    \end{equation}
we write $X\preccurlyeq^{(p)}Y$.
\end{definition}

\begin{lemma}[Cf.~\citet{Rusz1999, Ruszczynski2001}]\label{lem:1}
  With $F_X^{(1)}(\cdot)\coloneqq F_X(\cdot)$, the $k$\textsuperscript{th} ($k=2,3,\dots$) repeated integral is $F_X^{(k)}(x)\coloneqq\int_{-\infty}^xF_X^{(k-1)}(y)\,\mathrm{d}y$.
  The following two points are equivalent, they characterize stochastic dominance of \emph{integer} orders by repeated integrals:
  \begin{enumerate}
    \item $X\preccurlyeq^{(k)}Y$,
    \item\label{item:2} $F_Y^{(k)}(x)\ge F_X^{(k)}(x)$ for all $x\in\mathbb{R}$.
  \end{enumerate}
\end{lemma}
\begin{proof}
  It holds with \href{https://en.wikipedia.org/wiki/Cauchy_formula_for_repeated_integration}{Cauchy’s formula for repeated integration} that
  \begin{equation}
   	 F_X^{(k)}(x)=\frac{1}{(k-2)!}\int_{-\infty}^x(x-y)^{k-2}F_X(y)\,\mathrm{d}y.
  \end{equation}
  By integration by parts, the latter  is
  \begin{equation}
    F_X^{(k)}(x)=\frac{1}{(k-1)!}\int_{-\infty}^x(x-y)^{k-1}\,\mathrm{d}F_X(y),
  \end{equation}
  so that
  \[	F_X^{(k)}(x)=\frac{1}{(k-1)!}\int_{-\infty}^\infty(x-y)_+^{k-1}\mathrm{d}F_{X}(y)=\frac{1}{(k-1)!}\E(x-X)_+^{k-1},
  \]
  from which the assertion follows from the defining condition~\eqref{eq:domain} in Definition~\ref{def:Dominance}.
\end{proof}
\begin{remark}\label{rem:7}
  It follows from the iterated integral and~\ref{item:2} in Lemma~\ref{lem:1} that $X\preccurlyeq^{(k)} Y \implies X\preccurlyeq^{(k+1)}Y$ for all \emph{natural} numbers $k=1,2,\dots$.
  We notice next that
  \begin{equation}\label{eq:30}
    X\preccurlyeq^{(p)} Y \implies X\preccurlyeq^{(p^\prime)}Y\ \text{ for all \emph{real} numbers } 1\le p\le p^\prime \in \mathbb R.
  \end{equation}
  To this end note first that the characterization~\eqref{eq:29} is equivalent to
  \begin{equation}\label{eq:31}
    \int_{-\infty}^x(x-z)^{p-1}\mathrm dF_X(z)
    \ge \int_{-\infty}^x(x-z)^{p-1}\mathrm dF_Y(z)\text{ for all }x\in\mathbb R.
  \end{equation}
	With $\int_z^x(x-y)^{\alpha-1}(y-z)^{\beta-1}\mathrm dy =B(\alpha,\beta) (x-z)^{\beta+\alpha-1}$ ($B$ is Euler’s integral of the first kind) and integration by parts it follows that
  \begin{align}
		\int_{-\infty}^x(x-z)^{p^\prime-1}\, \mathrm d F_X(z)
    &=\frac1{B(p,p^\prime-p)}\int_{-\infty}^x\int_z^x(x-y)^{p^\prime-p-1}(y-z)^{p-1}\, \mathrm d y\,\mathrm d F_X(z)\\
    &=\frac1{B(p,p^\prime-p)}\int_{-\infty}^x(x-y)^{p^\prime-1-p}\int_{-\infty}^y(y-z)^{p-1}\,\mathrm d F_X(z)\,\mathrm d y\\
    &\ge\frac1{B(p,p^\prime-p)}\int_{-\infty}^x(x-y)^{p^\prime-1-p}\int_{-\infty}^x(y-z)^{p-1}\, \mathrm d F_Y(z)\,\mathrm d y\label{eq:32}\\
  &= \int_{-\infty}^x(x-z)^{p^\prime-1}\, \mathrm d F_Y(z),
 \end{align}
 where we have used the characterization~\eqref{eq:31} in~\eqref{eq:32}, as $x-y\ge 0$ and that $B(p,p^\prime-p)$ is well-defined and positive for $p^\prime>p$.
 The assertion again follows with~\eqref{eq:31}.
\end{remark}

\subsection{Comparison of stochastic order relations}
Different stochastic dominance relations may vary in strength (the implication~\eqref{eq:30} in the preceding  Remark~\ref{rem:7} is an example).
In what follows, we provide an explicit relation to compare stochastic dominance relations, which are built on different spectral functions.
\begin{proposition}[Comparison of spectral stochastic orders]\label{prop:Comp}
  Suppose that
  \begin{equation}\label{eq:33}
    \sigma_\mu(u) = \sigma(u)\cdot\int_0^{u_\beta}\frac{\mu(\mathrm{d}\beta)}{1-\beta}
  \end{equation}
  for some probability measure $\mu$, where $u_\beta$ is as defined in~\eqref{eq:Quantile}.
  Then the stochastic order associated with~$\sigma_\mu$  is weaker than the genuine stochastic order associated with~$\sigma$.
  Specifically, for different spectral functions $\sigma$ and $\sigma_\mu$, it holds that
  \[	X\preccurlyeq^{\|\cdot\|_{\sigma}}Y\implies X\preccurlyeq^{\|\cdot\|_{\sigma_{\mu}}}Y.
  \]
\end{proposition}
\begin{remark}
  The function $\sigma_\mu$ in~\eqref{eq:33} is indeed a spectral function.
  It is positive, as $\mu$ is a positive measure (thus~\ref{item:1} in Definition~\ref{def:Spectral}).
  The function is non-decreasing, as $u_\beta$ is non-decreasing for $\beta$ increasing.
  Finally, the function $\sigma_\mu$ is a density: indeed, it holds that
  \[	\int_0^1\sigma_\mu(u)\,\mathrm du
    = \int_0^1\sigma(u)\cdot \int_0^{u_\beta} \frac{\mu(\mathrm d\beta)}{1-\beta}\mathrm du
    = \int_0^1\int_{\beta_u}^1 \sigma(u)\,\mathrm du\ \frac{\mu(\mathrm d\beta)}{1-\beta}
    = \int_0^1 \mu(\mathrm d\beta) =1
  \]
  by integration by parts, where we have used the definition of $u_\beta$ in~\eqref{eq:Quantile}.
\end{remark}
\begin{proof}[Proof of Proposition~\ref{prop:Comp}]
  Since $x\preccurlyeq^{\|\cdot\|_\sigma}Y$, it holds with Theorem~\ref{thm:3} that $\mathcal{R}_{\sigma_\beta}(-X)\ge \mathcal{R}_{\sigma_\beta}(-Y)$ for all $\beta\in(0,1)$, where~$\sigma_\beta$ is defined in~\eqref{eq:sigma}.
  By the characterization~\eqref{eq:7}, this is
  \[	\int_{u_\beta}^1\frac{\sigma(u)}{1-\beta}F_{-X}^{-1}(u)\,\mathrm{d}u\ge\int_{u_\beta}^1\frac{\sigma(u)}{1-\beta}F_{-Y}^{-1}(u)\,\mathrm{d}u,\qquad\beta\in(0,1).
  \]
  Integrating the latter expression with respect to $\mu(\mathrm{d}\beta)$ establishes the inequality
  \[	\int_\beta^1\int_{u_{\beta^\prime}}^1\frac{\sigma(u)}{1-\beta^{\prime}}F_{-X}^{-1}(u)\,\mathrm{d}u\,\mu(\mathrm{d}\beta^{\prime})\ge \int_\beta^1\int_{u_{\beta^\prime}}^1\frac{\sigma(u)}{1-\beta^\prime}F_{-Y}^{-1}(u)\,\mathrm{d}u\,\mu(\mathrm{d}\beta^\prime),\qquad\beta\in(0,1).
  \]
  Interchanging the order of integration together with~\eqref{eq:18} gives that
  \[	\int_{u_\beta}^1\int_\beta^{\beta_u}\frac{\sigma(u)}{1-\beta^{\prime}}\mu(\mathrm{d}\beta^\prime)F_{-X}^{-1}(u)\,\mathrm{d}u
    \ge \int_{u_{\beta}}^{1}\int_\beta^{\beta_u}\frac{\sigma(u)}{1-\beta^\prime}\,\mu(\mathrm{d}\beta)F_{-Y}^{-1}(u)\,\mathrm{d}u, \qquad\beta\in(0,1),
  \]
  which in turn is
  \[	\int_{u_\beta}^1\sigma_\mu(u)F_{-X}^{-1}(u)\,\mathrm{d}u
    \ge \int_{u_\beta}^{1}\sigma_{\mu}(u)F_{-Y}^{-1}(u)\,\mathrm{d}u, \qquad\beta\in(0,1).
  \]
  This is the assertion.
\end{proof}

\section{Example: the expectile}\label{sec:Expectile}
The expectile risk measure, originally introduced by \citet{NeweyPowell}, has recently gained additional interest (cf.\ \citet{MalandiiUryasev}, \citet{Balbas2023} or \citet{SteinwartFarooq2018} for conditional regressions). A main reason for the additional interest in this risk measure is because it is the only elicitable risk functional (cf.\ \citet{Ziegel2014}).

As Proposition~\ref{prop:1} indicates, the higher order risk measure can be based on the dual norm. For this reason, the following section establishes the dual norm of expectiles first, as it is crucial in understanding its regret function in the risk quadrangle.
Next, we provide an explicit characterization of the higher order expectiles, that is, the higher order risk measure based on the expectile risk measure.
\medskip

The expectile is defined as a \emph{minimizer}. Its Kusuoka representation is central in elaborating the corresponding higher order risk functional.
\begin{definition}
  For $\alpha\in(0,1)$, the expectile is
  \begin{equation}\label{eq:34}
  e_\alpha(Y)=\argmin_{x\in\mathbb{R}}\E\ell_\alpha(Y-x),
\end{equation}
  where
  \begin{equation}\label{eq:35}
  \ell_{\alpha}(x)=\begin{cases}
  \quad \alpha\,x^2 & \text{if }x\ge0,\\
  (1-\alpha)x^2 & \text{else}
  \end{cases}
  \end{equation}
  is the asymmetric loss, or quadratic error function.
\end{definition}
The expectile satisfies the first order condition
\begin{equation}\label{eq:36}
	(1-\alpha)\E(x-Y)_+=\alpha\E(Y-x)_+,
\end{equation}
and $e_\alpha(\cdot)$ is a risk measure for $\alpha\in[\nicefrac12,1]$.
We mention that condition~\eqref{eq:36} provides a definition for $Y\in L^1$, it is thus more general than~\eqref{eq:34}, which requires $Y\in L^2$.
The Kusuoka representation of the expectile (cf.\ \citet[Proposition~9]{Bellini2014}) is given by
\begin{equation}\label{eq:37}
	e_\alpha(Y)=\max_{\gamma\in[0,1-\eta]}(1-\gamma)\cdot\E Y+\gamma\cdot\AVaR_{1-\frac{\gamma}{1-\gamma}\frac{\eta}{1-\eta}}(Y),
\end{equation}
where $\eta=\frac{1-\alpha}\alpha$, so that the risk level in~\eqref{eq:37} is $1-\frac{\gamma}{1-\gamma}\frac{\eta}{1-\eta}=\frac{\alpha(2-\gamma)-1}{(2\alpha-1)(1-\gamma)}$.
Involving spectral risk measures, the expectile can be recast as
\[e_\alpha(Y)=\sup\left\{ \mathcal{R}_{\sigma_\gamma}(Y)\colon\sigma_\gamma\in\mathcal{S}\right\} ,
\]
where $\mathcal{S}=\left\{ \sigma_\gamma\colon\gamma\in\left[0,1-\eta\right]\right\} $ collects the spectral functions
\[	s_\gamma(u)= \begin{cases}
		1-\gamma & \text{if }u\le1-\frac{\gamma}{1-\gamma}\frac{\eta}{1-\eta},\\
			\frac{1-\gamma}{\eta} & \text{else}.
\end{cases}
\]

The higher order expectile can be described by involving its dual norm (cf.~\eqref{eq:1.5}), as well as its Kusuoka representation (cf.\ Corollary~\ref{cor:mu}).
The following two (sub)sections elaborate these possibilities for the expectile.

\subsection{The dual norm of expectiles}
The higher order expectile can be described with the dual representation~\eqref{eq:6}, for which the dual norm of the expectile is necessary.

By the characterization of the loss function~\eqref{eq:36} it holds that $e_\alpha(Y)$ is well-defined for $Y\in L^1(P)$. This is enough to conclude that $\E|Y|\le C_\alpha\cdot e_\alpha(|Y|)$ for some constant $C_\alpha>0$ (\citet[Corollary~2.16]{LakshmananPichler2023} elaborate the tight bound $C_\alpha=\frac\alpha{1-\alpha}$).
It follows that $\mathcal Y^\ast= L^\infty$, so that $\|Z\|_\infty$ is well-defined for $Z\in \mathcal{Y}^*$.
\medskip

The following result provides the dual norm of the expectile explicitly.
\begin{proposition}[Dual norm of the expectile]\label{prop:10}
  For $\alpha\ge\nicefrac{1}{2}$, the dual norm is
  \begin{equation}\label{eq:38}
  \|Z\|_\alpha^*\coloneqq\sup\left\{ \E YZ\colon e_\alpha(|Y|)\le1\right\}
  \end{equation}
  (cf.~\eqref{eq:5}) . It holds that
  \begin{equation}\label{eq:39}
 				\|Z\|_\alpha^\ast  =\sup_{\beta\in(0,1)}(1-\beta)\cdot\AVaR_{\beta}(|Z|)+\beta\frac{1-\alpha}{\alpha}\|Z\|_\infty.
  \end{equation}
\end{proposition}
Notably, the norm $\|\cdot\|_\alpha^*$ is \emph{not} a risk measure itself, and~\eqref{eq:39} is \emph{not} a Kusuoka representation; indeed, the total weight in the representation~\eqref{eq:39} is
  \[	(1-\beta) + \beta\frac{1-\alpha}\alpha < 1
  \]
  for $\alpha\in (\nicefrac12,1]$.
\begin{proof}[Proof of Proposition~\ref{prop:10}]
  We may assume that $Z\ge 0$, as otherwise we may consider $\sign(Z)\cdot Y$
  instead of $Y$. For arbitrary sets $B$ and $G$  with  $B\subset G$ and $P(G)<1$ define the random variable
  \begin{equation}\label{eq:40}
    \tilde{Y}_{B,G}(\omega)\coloneqq\begin{cases}
    0 & \text{if }\omega\in B,\\
    1 & \text{if }\omega\in G\setminus B,\text{ and}\\
    \frac{1-\alpha}{\alpha}\cdot\frac{P(B)}{1-P(G)}+1 & \text{else}.
    \end{cases}
  \end{equation}
  Note, that
  \[	(1-\alpha)\cdot P(B)(1-0)= \alpha\cdot\bigl(1-P(G)\bigr)\left(\frac{(1-\alpha)P(B)}{\alpha(1-P(G))}+1-1\right),
  \]
  and hence $e_\alpha(\tilde{Y}_{B,G})=1$ by the defining equation~\eqref{eq:36}. It follows with~\eqref{eq:38} that
  \[	\|Z\|_\alpha^*\ge\E Z\,Y_{B,G}.
  \]
  As $B\subset G$ are arbitrary, we conclude in particular that
  \[	\|Z\|_\alpha^*\ge \bigl((1-P(B)\bigr)\cdot\AVaR_{P(B)}(Z)+P(B)\frac{1-\alpha}{\alpha}\cdot\AVaR_{P(G)}(Z),
  \]
  because the random variables
  \[	\tilde Y_{B,G}= \bigl(1-P(B)\bigr)\cdot\frac{1}{1-P(B)}\one_{[P(B),1]}(U)+P(B)\frac{1-\alpha}{\alpha}\cdot\frac{1}{1-P(G)}\one_{[P(G),1]}(U)
  \]
  satisfy all conditions from above for any uniform variable $U$. Now let $P(G)\to1$ and by denoting $\beta=P(B)$ it follows that
  \[	\|Z\|_\alpha^*\ge \sup_{\beta\in(0,1)}(1-\beta)\cdot \AVaR_\beta(Z)+\beta\frac{1-\alpha}{\alpha}\esssup Z,
  \]
  as $\AVaR_{\gamma}(Z)\to\esssup Z$ for $\gamma\to1$.

  As for the converse observe that we may assume $e_\alpha(Y)=1$ for the optimal random variable in~\eqref{eq:38}. Consider the Lagrangian
  \begin{equation}\label{eq:41}
  L(Y;\lambda,\mu)\coloneqq\E ZY-\lambda\bigl((1-\alpha)\E(1-Y)_+-\alpha\E(Y-1)_{+}\bigr)+\E\mu Y,
  \end{equation}
	where the Lagrangian multiplier $\lambda\in\mathbb{R}$ is associated with the equality constraint $e_\alpha(Y)=1$, i.e.,~\eqref{eq:36}, and the measurable variable $\mu\in L^1$, $\mu\ge0$, is associated with the inequality constraint $Y\ge 0$.
  Provided That the derivative exists, the first order conditions are
  \begin{align}
  0 & =\frac\partial{\partial Y}L(Y;\lambda,\mu),\nonumber \\
  \shortintertext{or}Z & =\lambda\bigl(-(1-\alpha)\one_{\{Y<1\}}-\alpha\,\one_{\{Y>1\}}\bigr)-\mu\cdot\one_{\{Y=0\}}.\label{eq:42}
  \end{align}
  Now note that the left-hand side of~\eqref{eq:42} involves the \emph{variable}~$Z$, while the right-hand side only involves \emph{constants}, except on $\{Y=0\}$, where~$\mu$ is not necessarily constant.
  The first order conditions~\eqref{eq:42} thus hold true on plateaus of $Z$, if they coincide with $\{Y<1\}$ or $\{Y>1\}$; for $\{Y=0\}$, equation~\eqref{eq:42} is $\mu=-Z-\lambda(1-\alpha)$; for $\{Y=1\}$, the derivative of~\eqref{eq:41} does not exist or depends on the direction.

  It follows, that the optimal~$Y$ in~\eqref{eq:38} exactly is of form~\eqref{eq:40} and hence the assertion.
\end{proof}

\subsection{Higher order expectiles}
The Kusuoka representation~\eqref{eq:37} is the basis for the expectile’s higher order risk measure.
\begin{proposition}
  For $\beta\in(0,1)$, the higher order expectile is
  \begin{equation}\label{eq:43}
    \bigl(e_\alpha\bigr)_\beta(Y)=\max_{\gamma\in [0,1-\eta]} \begin{cases}
      \left(1-\frac{\gamma}{1-\beta}\right)\AVaR_{\frac{\beta}{1-\gamma}}(Y)+\frac{\gamma}{1-\beta}\AVaR_{1-\frac{\gamma}{1-\gamma}\frac{\eta}{1-\eta}}(Y) & \text{if }\frac{\gamma}{1-\beta}<1-\eta,\\
      \AVaR_{1-(1-\beta)(1-\tilde\alpha)}(Y) & \text{else},
    \end{cases}
  \end{equation}
	where $\eta=\frac{1-\alpha}\alpha$ (as above) and $\tilde \alpha\coloneqq1-\frac{\eta}{1-\gamma}$.
\end{proposition}
\begin{proof}
  The measure in the Kusuoka representation~\eqref{eq:37} is $\mu(\cdot)=(1-\gamma)\delta_{0}+\gamma\cdot\delta_{1-\frac{\gamma}{1-\gamma}\frac{\eta}{1-\eta}}$.
  To apply Corollary~\ref{cor:AVaR} we set $p_{1}\coloneqq1-\gamma$ and $p_2=\gamma$, the corresponding risk levels are $\alpha_1=0$ and $\alpha_2 = 1-\frac{\gamma}{1-\gamma}\frac{\eta}{1-\eta}$.
  The mixed risk level is $\tilde\alpha\coloneqq \frac{\alpha_1\frac{p_1}{1-\alpha_1}+\alpha_2\frac{p_2}{1-\alpha_2}}{\frac{p_1}{1-\alpha_1}+\frac{p_2}{1-\alpha_2}}=\frac{\alpha(2-\gamma)-1}{\alpha(1-\gamma)}=1-\frac{\eta}{1-\gamma}$.

  We distinguish the cases $\frac{\gamma}{1-\beta}<1-\eta$ and $\frac{\gamma}{1-\beta}<1-\eta$, which are equivalent to $u_\beta\lessgtr \alpha_2$, i.e., $1-\frac{1-\beta}{\frac{\gamma}{1-\alpha_1}+\frac{1-\gamma}{1-\alpha_2}}\lessgtr\alpha_2$ in view of~\eqref{eq:19}.
  In the first case, the critical equation~\eqref{eq:17} is $(1-\gamma)u_{\beta}=\beta$, while it is $(1-\gamma)u_\beta+\gamma\frac{u_\beta-\alpha_2}{1-\alpha_2}=\beta$ in the other case; the solutions thus are $u_\beta=\frac{\beta}{1-\gamma}$ and $u_\beta=\frac{\alpha(2-\beta-\gamma)-1+\beta}{\alpha(1-\gamma)}$.
  The corresponding weights $p_0$ (cf.~\eqref{eq:17} again) are $p_0=\frac{1-u_\beta}{1-\beta}(1-\gamma)$, or $p_0=\frac{1-u_\beta}{1-\beta}\left(\frac{1-\gamma}{1-0}+\frac{\gamma}{1-\alpha_2}\right)=1$.
  Finally, note that $u_\beta= 1-(1-\beta)(1-\tilde{\alpha})$.

  The assertion follows with~\eqref{eq:18} and~\eqref{eq:17-2} in Corollary~\ref{cor:AVaR}.
\end{proof}
The average value-at-risk is ‘closed under higher orders’, as its higher order variant is an average value-at-risk as well (cf.~\eqref{eq:14}).
This is not the case for the expectile, as the first term in~\eqref{eq:43} is not an expectation as in the genuine Kusuoka representation~\eqref{eq:37}.
Repeating the construction and passing to higher order expectiles leads to more complicated risk measures.

\section{Summary}\label{sec:Summary}
Higher order risk measures naturally integrate with stochastic optimization, as they are stochastic optimization problems themselves.
This paper presents and derives explicit forms of higher order risk measures, specifically for spectral risk measures. These risk measures constitute the central building block of general law invariant risk measures.

Extending these results result it is demonstrated that stochastic dominance relations can be characterized by employing higher order risk measures, and vice versa.
We provide a verification theorem, which makes higher stochastic dominance relations accisible to numerical compuations.

The results are exemplified for expectiles, a specific risk measure with unique properties.

\bibliographystyle{abbrvnat}
\bibliography{LiteraturAlois}

\end{document}